\documentclass[a4paper,11pt]{article}

\usepackage{tcs}
\makeatletter
\newcommand{\fallbackcref}[1]{\expandafter\fallbackcref@dispatch#1:\@nil{#1}}
\def\fallbackcref@dispatch#1:#2\@nil#3{%
    \@ifundefined{fallbackcref@#1}{\ref{#3}}{\csname fallbackcref@#1\endcsname{#3}}%
}
\def\fallbackcref@thm#1{Theorem~\ref{#1}}
\def\fallbackcref@theorem#1{Theorem~\ref{#1}}
\def\fallbackcref@conj#1{Conjecture~\ref{#1}}
\def\fallbackcref@lemma#1{Lemma~\ref{#1}}
\def\fallbackcref@lem#1{Lemma~\ref{#1}}
\def\fallbackcref@prop#1{Proposition~\ref{#1}}
\def\fallbackcref@cor#1{Corollary~\ref{#1}}
\def\fallbackcref@def#1{Definition~\ref{#1}}
\def\fallbackcref@eq#1{Equation~\eqref{#1}}
\def\fallbackcref@sec#1{Section~\ref{#1}}
\newcommand{\installfallbackcref}{%
    \renewcommand{\cref}[1]{\fallbackcref{##1}}%
    \renewcommand{\Cref}[1]{\fallbackcref{##1}}%
}
\makeatother

\IfFileExists{cleveref.sty}{
    \crefname{theorem}{Theorem}{Theorems}
    \Crefname{theorem}{Theorem}{Theorems}
    \crefname{thm}{Theorem}{Theorems}
    \Crefname{thm}{Theorem}{Theorems}
    \crefname{conjecture}{Conjecture}{Conjectures}
    \Crefname{conjecture}{Conjecture}{Conjectures}
    \crefname{definition}{Definition}{Definitions}
    \Crefname{definition}{Definition}{Definitions}
    \crefname{lemma}{Lemma}{Lemmas}
    \Crefname{lemma}{Lemma}{Lemmas}
    \crefname{proposition}{Proposition}{Propositions}
    \Crefname{proposition}{Proposition}{Propositions}
    \crefname{corollary}{Corollary}{Corollaries}
    \Crefname{corollary}{Corollary}{Corollaries}
    \crefname{equation}{Equation}{Equations}
    \Crefname{equation}{Equation}{Equations}
}{
    \installfallbackcref
}

\newcommand{\val}{\operatorname{val}}
\newcommand{\OPT}{\operatorname{OPT}}
\newcommand{\SDP}{\operatorname{SDP}}
\newcommand{\one}{\boldsymbol{1}}
\newcommand{\Normal}{\mathcal{N}}

\begin{document}

\title{On the Approximability of Boolean Max-$k$-CSP}
\date{}
\author{
Ainesh Bakshi\\
\texttt{ainesh@nyu.edu}\\
NYU
}

\maketitle
\thispagestyle{empty}

\begin{abstract}
Consider the problem of maximizing the number of satisfied constraints of an arbitrary boolean constraint satisfaction problem with arity $k$. We obtain a polynomial time algorithm that achieves a $(k/2^k)$-approximation, improving on the previous best guarantee of $0.626612\; k/2^k$, due to Makarychev and Makarychev~\cite{makarychev2012approximation}. 
Assuming the Unique Games Conjecture, De and Mossel showed that achieving an approximation ratio better than $(k+1)/2^k$ for odd $k$ and $(k+2)/2^k$ for even $k$, is NP-hard~\cite{de2013explicit}. The main technical ingredient is an extension of a recently established Gaussian comparison inequality, used to resolve the Weak Simplex Conjecture in coding theory~\cite{Mulgund2026GaussianMaxima}.  
\end{abstract}

\newpage
\setcounter{page}{1}

\section{Introduction}

Constraint satisfaction problems (CSPs) provide a unifying framework for several fundamental problems in theoretical computer science. Given a collection of variables and local constraints, the goal is to find an assignment to the variables that satisfies as many constraints as possible. This simple framework captures several fundamental problems, including SAT, Max-Cut, Graph Coloring, and many others. Since their introduction, CSPs have repeatedly served as a testbed for new algorithmic ideas and lower bound techniques.

On the hardness side, the PCP theorem showed that it is NP-hard to distinguish satisfiable CSP instances from those in which every assignment violates a constant fraction of the constraints~\cite{arora1998proof}. This was sharpened further in Khot's work on the Unique Games Conjecture, which postulates that it is NP-hard to disinguish almost satisfiable instances of a particularly structured CSP from instances that are far from satisfiable~\cite{khot2002power}. On the algorithms side, central tools such as semidefinite programs (SDPs), their extension to convex hierarchies, and randomized rounding were developed to analyze CSPs~\cite{goemans1995improved,barak2011rounding}.

A landmark theorem of Raghavendra brought remarkable unity to the approximability of CSPs~\cite{raghavendra2008optimal}. Assuming the Unique Games Conjecture, a generic semidefinite-programming framework achieves the optimal approximation ratio for every finite-domain CSP. Thus, in principle, the approximability of every CSP is characterized by an SDP. Raghavendra's result is nevertheless existential, i.e. for a given CSP, it does not identify the optimal approximation ratio or provide an explicit optimal rounding procedure. Determining the optimal approximation ratio and rounding algorithm remains a difficult, problem-specific challenge. 

We focus on the problem of maximizing the number of satisfied constraints in an arbitrary boolean CSP where each constraint has arity at most $k$, i.e. each variable is a bit and each constraint is an arbitrary predicate on at most $k$ variables. 
In the worst case, a constraint may accept only one of its $2^k$ local assignments, so a random assignment gives a $2^{-k}$ approximation. Charikar, Makarychev, and Makarychev first showed that one can always gain a factor of $\Omega(k)$, obtaining a $0.44 k/2^k$ approximation~\cite{charikar2009near}, and subsequently, Makarychev and Makarychev improved the guarantee to $(0.626612-o_k(1))k/2^k$~\cite{makarychev2012approximation}.

Complementing these algorithmic results, Samorodnitsky and Trevisan first established $\Theta(k/2^k)$ as the right scale under UGC~\cite{samorodnitsky2006gowers}. Austrin and Mossel sharpened the bound to $(1+o_k(1))k/2^k$~\cite{austrin2009approximation}, and the best known hardness result is due to De and Mossel who proved that doing better than $(k+1)/2^k$ for odd $k$ and $(k+2)/2^k$ for even $k$ is NP-hard under UGC~\cite{de2013explicit}. Finally, Chan unconditionally proved that it is NP-hard to achieve a ratio $2k/2^k+\varepsilon$~\cite{chan2016approximation}. Taken together, these results establish that the correct scale is $\Theta(k/2^k)$ and suggest that the optimal approximation ratio is $k/2^k$, up to a $1+o(1)$ factor. 
Writing the optimal approximation ratio as $c_k k/2^k$, Makarychev and Makarychev asked whether $c_k$ converges and, if so, to determine its limit~\cite{makarychev2017approximation}. They conjectured that $c_k\to1$ and, more specifically, that their algorithm from~\cite{makarychev2012approximation} already achieves a $(1-o_k(1))k/2^k$ approximation. We prove this conjecture and obtain the following stronger guarantee:

\begin{theorem}[Main theorem]\label{thm:main}
For every fixed integer $k\geqslant10$ and every $\varepsilon>0$, there is a randomized algorithm that, given any Boolean Max-$k$-CSP instance $\calI$ on $n$ variables with $m$ constraints, outputs an assignment $x$ satisfying
\[
    \E[\val_{\calI}(x)]
    \geqslant 
    \Paren{k/2^{k}-\varepsilon}\OPT(\calI)\,,
\]
and runs in $\poly(n, m,k, \log(1/\varepsilon))$ time.
\end{theorem}

In exact arithmetic, our rounding achieves the factor $k/2^k$, and the arbitrarily small loss in \cref{thm:main} comes only from finite-precision implementation. Combined with the hardness result of De and Mossel, we can conclude that $\lim_{k\to \infty} c_k = 1$ and the SDP from~\cite{makarychev2012approximation} admits a rounding that achieves an approximation ratio of $k/2^k$, avoiding any $(1+o(1))$-factor loss.

We use the same SDP relaxation as Makarychev and Makarychev~\cite{makarychev2012approximation}. After the standard conjunction reduction and padding, every constraint is a weighted conjunction of exactly $k$ literals. For each variable $x_i$, let $d_i$ be its SDP decision direction. We use a noisy Gaussian hyperplane rounding that interpolates between the SDP rounding and an independent random assignment. In particular, $x_i$ is determined by the sign of $\sqrt{\delta}\,\langle g,d_i\rangle+\sqrt{1-\delta}\,\xi_i,$, where $\delta=\frac{\pi}{k}$, $g$ is a common standard Gaussian vector and the $\xi_i$'s are independent standard Gaussians. Thus, the SDP hyperplane score contributes a $\delta$-fraction of the variance, and the rest comes from independent noise. To analyze a fixed conjunction, we orient each decision direction, and hence its score, toward the bit required by that conjunction. The conjunction is satisfied exactly when all $k$ oriented scores are nonnegative.

The SDP constraints force the covariance of the oriented hyperplane scores to have the form $qJ+(1-q)G$, where $q$ is the clause's SDP value, $J$ is the all ones matrix, and $G$ is an arbitrary correlation matrix. After adding noise, the covariance becomes $R=\delta(qJ+(1-q)G)+(1-\delta)I$, and hence $ R\succeq\Paren{\delta q+\frac{1-\delta}{k}}J.$
Our main technical ingredient is a Gaussian orthant amplification inequality tailored to precisely this situation (see \cref{sec:gaussian}). It shows that if a centered Gaussian vector $X\in\R^k$ has correlation matrix $R\succeq\eta J$, then $\Pr[X\geqslant0]\geqslant2^{-k}\exp((k-\eta^{-1})/\pi)$, regardless of the remaining correlations. With $\delta=\pi/k$, an elementary scalar inequality turns this bound into $\Pr[C\text{ is satisfied}]\geqslant(k/2^k)q$. Summing over clauses proves the theorem. The orthant inequality itself follows by centering exponential tilts of the half-line indicators, applying a  centered Gaussian product inequality due to Mulgund~\cite{Mulgund2026GaussianMaxima}, and evaluating the resulting variational bound at a carefully chosen test point.

\section{Preliminaries}\label{sec:preliminaries}

\paragraph{Constraint satisfaction problems.}
A boolean $k$-CSP instance $\calI$ consists of $n$ variables $x_1,\ldots,x_n\in\zo$ and $m$ weighted clauses $C_1,\ldots,C_m$. Each clause $C_j$ has a support $S_j\subseteq[n]$ of size at most $k$, a predicate $P_j\from\zo^{S_j}\to\zo$, and a nonnegative weight $\lambda_j$. We normalize the weights so that $\sum_{j=1}^m\lambda_j=1$.
For an assignment $x\in\zo^n$, define
\begin{equation*}
    \val_{\calI}(x)
    =\sum_{j=1}^m\lambda_jP_j\Paren{x|_{S_j}},
    \qquad
    \OPT(\calI)=\max_{x\in\zo^n}\val_{\calI}(x).
\end{equation*}
When the instance is fixed, we abbreviate these quantities by $\val(x)$ and $\OPT$. An algorithm is an $\alpha$-approximation if it outputs an assignment $x$ satisfying $\E[\val(x)]\geqslant\alpha\OPT$.

\paragraph{Vectors and matrices.}
For $m\in\N$, let $[m]=\{1,\ldots,m\}$. We use $\langle u,v\rangle$ and $\norm{v}$ for the Euclidean inner product and norm. For a matrix $A$, we use $\norm{A}_{\mathrm{op}}$ and $\norm{A}_F$ to denote operator and Frobenius norms. The identity matrix is denoted by $I$, the all-ones vector is $\one$, and $J=\one\one^\top$ is the all-ones matrix. For symmetric matrices $A$ and $B$, the notation $A\succeq B$ means that $A-B$ is positive semidefinite.
Inequalities between vectors are coordinatewise. Thus, for $v,w\in\R^m$, the relation $v\leqslant w$ means $v_i\leqslant w_i$ for every $i\in[m]$. In particular, $X\leqslant c\one$ denotes the event that $X_i\leqslant c$ for every coordinate $i$.

\paragraph{Probability.}
For a measurable set $A\subseteq\R^m$, its standard Gaussian measure is
\[
    \gamma_m(A)
    =\frac{1}{(2\pi)^{m/2}}
      \int_A\exp\Paren{-\frac{\norm{x}^2}{2}}\diff x.
\]
Equivalently, $\gamma_m(A)=\Pr[g\in A]$ for $g\sim\Normal(0,I_m)$.
For $Z\sim\Normal(0,1)$, let $\Phi(t)=\Pr[Z\leqslant t]$ denote the standard Gaussian cumulative distribution function. More generally, $X\sim\Normal(0,R)$ denotes a centered Gaussian vector with covariance matrix $R\succeq 0$.

\section{Background and Warm Up}\label{sec:overview}

In this section, we describe a $(1-o(1))k/2^k$-approximation for boolean Max-$k$-CSP. This already resolves the conjecture of Makarychev and Makarychev~\cite{makarychev2017approximation} and exposits the connection to the new Gaussian comaprison inequality used to resolve the weak simplex conjecture in coding theory~\cite{Mulgund2026GaussianMaxima}. The algorithm uses the natural semidefinite relaxation of Makarychev and Makarychev~\cite{makarychev2012approximation}. Gaussian hyperplane rounding of this relaxation reduces clause satisfaction to Gaussian orthant probabilities. We use a Gaussian comparison inequality~\cite{Mulgund2026GaussianMaxima} as a black box to lower-bound these probabilities, and combine hyperplane rounding with a random assignment to obtain the stated ratio. We then extend Mulgund's argument to remove the $o(1)$ loss.

We first apply Trevisan's approximation-preserving reduction~\cite{trevisan1998parallel}, replacing each constraint by at most $2^k$ conjunctions, one for each accepting local assignment. Consider a constraint $P_j$ on the variables indexed by $S_j$. For every accepting assignment $a\in\zo^{S_j}$, the corresponding conjunction $C_{j,a}$ requires $x_i=a_i$ for every $i\in S_j$ and inherits the weight $\lambda_j$. For any global assignment $x$, exactly one local assignment agrees with $x|_{S_j}$. If $P_j$ accepts this assignment, exactly one of the new conjunctions is satisfied, and if $P_j$ rejects, no conjunction is satisfied. The new conjunctions therefore contribute exactly $\lambda_jP_j(x|_{S_j})$, preserving the value of every assignment. For each conjunction $C$, let $\supp(C)\subseteq[n]$ be its support and let $b_C\in\zo^{\supp(C)}$ denote its unique satisfying local assignment. Thus, $C$ is satisfied exactly when $x_i=b_C(i)$ for every $i\in\supp(C)$.

\paragraph{SDP setup and relaxation.}
The Makarychev--Makarychev relaxation is easiest to construct from an integral assignment. Fix a unit vector $e$. For every literal $x_i=b$, introduce a global label vector $v_{i,b}$, and given an assignment $x$, set
\[
    v_{i,b}=\begin{cases}
        e,&x_i=b,\\
        0,&x_i\neq b
    \end{cases} \,, \qquad \textrm{ which captures the event $x_i=b$.}
\]
Thus, $\norm{v_{i,b}}^2$ is the indicator that the literal $x_i=b$ is true. For each variable $x_i$, the two literal vectors $v_{i,0}$ and $v_{i,1}$ are orthogonal because the two values are mutually exclusive, and $\norm{v_{i,0}}^2 + \norm{v_{i,1}}^2 = 1$, since exactly one value is chosen.

The clause vector is the analogue of the literal vector for an entire conjunction. For each conjunction $C$, introduce a local vector
\[
    z_C=\begin{cases}
        e,&C\text{ is satisfied},\\
        0,&C\text{ is not satisfied}.
    \end{cases}
\]
The vector is local because every conjunction receives its own $z_C$. Its squared norm is the indicator of the joint event that all requirements of $C$ hold simultaneously. Since satisfying $C$ makes every required literal true and is incompatible with every opposite literal, the integral vectors satisfy
\[
    \langle v_{i,b_C(i)},z_C\rangle=\norm{z_C}^2,
    \qquad
    \langle v_{i,1-b_C(i)},z_C\rangle=0
    \qquad\text{for every }i\in\supp(C).
\]
For example, consider $C=(x_1=1)\wedge(x_2=0)\wedge(x_4=1).$
If $C$ is satisfied, then the required literal vectors $v_{1,1}$, $v_{2,0}$, and $v_{4,1}$ all equal $z_C$, while the conflicting literal vectors $v_{1,0}$, $v_{2,1}$, and $v_{4,0}$ are zero. Therefore, $ \langle v_{1,1},z_C\rangle
    =\langle v_{2,0},z_C\rangle
    =\langle v_{4,1},z_C\rangle
    =\norm{z_C}^2$,
whereas $\langle v_{1,0},z_C\rangle=\langle v_{2,1},z_C\rangle=\langle v_{4,0},z_C\rangle=0$. If $C$ is not satisfied, $z_C=0$ and the same identities hold trivially.

To obtain the SDP relaxation, we drop the requirement that every vector equal either $0$ or $e$, while retaining the constraints. We relax the exact label mass $\norm{v_{i,0}}^2+\norm{v_{i,1}}^2=1$ to $\norm{v_{i,0}}^2+\norm{v_{i,1}}^2\leqslant1$, allowing some mass to remain unassigned. The squared norm $\norm{z_C}^2$ is the fractional score of clause $C$, so the SDP is
\begin{align*}
    \textup{maximize}\qquad
        &\sum_C\lambda_C\norm{z_C}^2 \\
    \textup{subject to}\qquad
        &\norm{v_{i,0}}^2+\norm{v_{i,1}}^2\leqslant1
            &&\text{for every }i\in[n],\\
        &\langle v_{i,0},v_{i,1}\rangle=0
            &&\text{for every }i\in[n],\\
        &\langle v_{i,b_C(i)},z_C\rangle=\norm{z_C}^2
            &&\text{for every }C\text{ and }i\in\supp(C),\\
        &\langle v_{i,1-b_C(i)},z_C\rangle=0
            &&\text{for every }C\text{ and }i\in\supp(C),\\
        &\norm{z_C}^2\leqslant1
            &&\text{for every }C.
\end{align*}
Every integral assignment embeds as above with objective equal to its value, so $\SDP\geqslant\OPT$.

\paragraph{Gaussian hyperplane rounding.}
Gaussian rounding assigns each variable according to which side of a random hyperplane its decision vector lies. The SDP gives two literal vectors for $x_i$, one for each possible value, so the natural decision direction is $d_i=v_{i,1}-v_{i,0}$.
Indeed, for a Gaussian normal vector $g$, the sign of $\langle g,d_i\rangle$ compares the Gaussian scores of the two values and we can wlog assume that $\norm{d_i}=1$. We now sample a standard Gaussian vector $g$ and set $ x_i=\boldsymbol{1}\Braces{\langle g,d_i\rangle\geqslant0}.$
The expected value of the rounded assignment is
\[
    \E[\val(x)]
    =\sum_C\lambda_C\Pr[C\text{ is satisfied}].
\]
It therefore suffices to analyze the probability of a clause being satisfied, one clause at a time. In particular, if the rounding satisfies every clause $C$ with probability at least $\alpha\norm{z_C}^2$, then $\E[\val(x)]
    \geqslant\alpha\sum_C\lambda_C\norm{z_C}^2
    =\alpha\SDP
    \geqslant\alpha\OPT.$

\paragraph{Analyzing a single clause.}
We now fix a clause $C$, and drop the subscripts for convenience. Recall that satisfying $C$ is equivalent to $x_i=b(i)$ for every $i\in\supp(C)$. We therefore orient each decision vector toward the value required by the clause: let $w_i=d_i$ if $b(i)=1$, and let $w_i=-d_i$ otherwise. After this orientation,
\[
    C\text{ is satisfied}
    \quad\Longleftrightarrow\quad
    \langle g,w_i\rangle\geqslant0
    \quad\text{for every }i\in\supp(C)\,,
\]
so the probability of satisfying $C$ is a Gaussian orthant probability. If the $w_i$'s were arbitrary, this probability could range from $0$, if two vectors point in opposite directions, to $1/2$, if all the vectors are identical. The SDP constraints provide the additional structure we need:
\[
    \langle w_i,z\rangle
    =\langle v_{i,b(i)},z\rangle
     -\langle v_{i,1-b(i)},z\rangle
    =\norm{z}^2.
\]
Thus, the projection of every $w_i$ onto $z$ is exactly $z$. Since the $w_i$'s are unit vectors, we may write
\[
    w_i=z+\sqrt{1-\norm{z}^2}\,q_i,
    \qquad q_i\perp z,
    \qquad \norm{q_i}=1.
\]
Let $G$ be the Gram matrix of the residual vectors $q_i$. The SDP imposes no constraints on their pairwise inner products, so $G$ may be an arbitrary correlation matrix. Let $\Sigma$ be the Gram matrix of the $w_i$'s, and thus
\begin{equation}\label{eq:overview-clause-covariance}
    \Sigma=\norm{z}^2J+(1-\norm{z}^2)G.
\end{equation}
Let $X\sim\Normal(0,\Sigma)$ and observe that $\Pr[C\text{ is satisfied}]
    =\Pr[X\geqslant0]$.
We have therefore reduced the analysis of the clause to the following question: how small can this orthant probability be when the covariance contains a common rank-one component $\norm{z}^2J$ and an otherwise arbitrary correlation matrix $G$?

\paragraph{Gaussian comparison inequalities.}
The Gaussian comparison inequality of Mulgund~\cite[Theorem~2.1]{Mulgund2026GaussianMaxima} says that if $R$ is a $k\times k$ correlation matrix, for any $Y\sim\Normal(0,R)$ and every $c\in\R$,
\begin{equation}\label{eq:overview-mulgund}
    \textrm{if } R-J/k\succeq0\,,
    \qquad
    \Pr[Y\leqslant c\one]\geqslant\Phi(c)^k\,,
\end{equation}
where $\Phi(c)$ is the cumulative distribution function. We first consider the case where $1/k\leqslant\norm{z}^2$, and in this range,
\[
    \Sigma-\frac1kJ
    =\Paren{\norm{z}^2-\frac1k}J
      +(1-\norm{z}^2)G
    \succeq0,
\]
so we can apply \cref{eq:overview-mulgund} directly to $\Sigma$ with $c=0$. However, this gives only $\Pr[X\geqslant0]\geqslant2^{-k}$, whereas the desired clausewise bound is $(k/2^k)\norm{z}^2$. When $\norm{z}^2=1$, we have $\Sigma=J$ and the orthant probability is $1/2$, so we assume $\norm{z}^2<1$.
The crucial idea is to only use $1/k$ weight along $J$, which is exactly what we need to apply \cref{eq:overview-mulgund}.  To this end, we define $Q=\frac1kJ+\frac{k-1}{k}G.$
Thus, $Q$ lies on the same line between $G$ and $J$ as $\Sigma$, and satisfies $Q-\frac1kJ=\frac{k-1}{k}G\succeq0.$
We pick $\rho =\frac{k\norm{z}^2-1}{k-1}$ and rewrite $\Sigma$ as follows:
\begin{equation}\label{eq:overview-common-noise}
\begin{aligned}
    \Sigma
    &=(1-\rho)Q+\rho J =\Paren{\frac{1-\rho}{k}+\rho}J
      +(1-\rho)\frac{k-1}{k}G =\norm{z}^2J+(1-\norm{z}^2)G.
\end{aligned}
\end{equation}
Probabilistically, this decomposes $X = \sqrt{1-\rho}\,U+\sqrt\rho\,Z\one$, where $U\sim\Normal(0,Q)$ and $Z$ is an independent standard Gaussian.
Let $a=\sqrt{\rho/(1-\rho)}$ and for a fixed $s\in\R$ condition on $Z=s$. Since $Z$ appears identically in every coordinate,
\begin{align*}
    \Pr[X\geqslant0\mid Z=s] =\Pr\Bracks{
        \sqrt{1-\rho}\,U+\sqrt\rho\,s\one\geqslant0
    } =\Pr[U\geqslant-as\one] =\Pr[U\leqslant as\one],
\end{align*}
where the last equality follows from symmetry. Since \cref{eq:overview-mulgund} applies for any $c$, $\Pr[X\geqslant0\mid Z=s]
    \geqslant\Phi(as)^k.$
Finally, averaging the conditional probability over $Z$, we have
\begin{equation}\label{eq:overview-pk}
    \Pr[X\geqslant0]
    =\E_Z\Pr[X\geqslant0\mid Z]
    \geqslant
    \E_Z\Phi\Paren{
        \sqrt{\frac{k\norm{z}^2-1}{k(1-\norm{z}^2)}}Z
    }^k.
\end{equation}

\paragraph{A one-dimensional Gaussian estimate.}
The main technical contribution here is to show that for $\eta_k=\frac{2\pi\log k+1}{k}$, we can lower bound the expression in \cref{eq:overview-pk} by the desired $k^2/2^k$ whenever $\norm{z}^2\geqslant(1+\eta_k)/k$.
It remains to lower bound the following one-dimensional quantity
\[
    P_k(a)=\E_Z\Phi(aZ)^k,
    \qquad \textrm{where } \qquad
    a^2= (\norm{z}^2-1/k)/ (1-\norm{z}^2).
\]
To this end, let $G$ be a standard Gaussian and let $\mu=\E|G|=\sqrt{2/\pi}$. By symmetry and completing the square, for every $x\in\R$,
\begin{align*}
    \E e^{x|G|}
    &=\frac{2}{\sqrt{2\pi}}
      \int_0^\infty
      \exp\Paren{xt-\frac{t^2}{2}}\diff t
      =2e^{x^2/2}\Phi(x).
\end{align*}
For each fixed $x$, the function $y\mapsto e^{xy}$ is convex, even when $x<0$. Jensen's inequality therefore gives $\E e^{x|G|}\geqslant e^{x\E|G|}=e^{\mu x}$, and we have $\Phi(x) \geqslant \frac12\exp\Paren{\mu x-\frac{x^2}{2}}$.
Using this inequality with $x=aZ$, raising it to the $k$-th power and averaging over $Z$ gives
\begin{align*}
    P_k(a)
    &\geqslant
    2^{-k}\E_Z\exp\Paren{
        k\mu aZ-\frac{ka^2}{2}Z^2
    } =\frac{2^{-k}}{\sqrt{1+ka^2}}
      \exp\Paren{
          \frac{k^2a^2}{\pi(1+ka^2)}
      },
\end{align*}
where the last equality follows by completing the square and using $\mu^2=2/\pi$. For our choice $\eta_k=(2\pi\log k+1)/k$, at the endpoint $\norm{z}^2=(1+\eta_k)/k$ we have $ a^2=\frac{\eta_k}{k-1-\eta_k}$ and $ka^2=o(1)$. Therefore, the
\[
    \frac{k^2a^2}{\pi(1+ka^2)}
    =\frac{k^2\eta_k}{\pi(k-1)(1+\eta_k)}
    =2\log k+\frac1\pi+o(1).
\]
The one-dimensional estimate therefore gives
\[
    P_k(a)
    \geqslant
    \frac{2^{-k}}{\sqrt{1+ka^2}}
    \exp\Paren{2\log k+\frac1\pi+o(1)}
    =\frac{k^2}{2^k}\exp\Paren{\frac1\pi+o(1)}
    \geqslant
    \frac{k^2}{2^k}
\]
for all sufficiently large $k$.
The same bound continues to hold for larger values of $\norm{z}^2$. Indeed, pairing $Z=s$ with $Z=-s$ turns the corresponding contribution to $P_k(a)$ into $p^k+(1-p)^k$, where $p=\Phi(as)\geqslant1/2$. This expression is nondecreasing in $p$, and $p$ is nondecreasing in $a$. Thus, $P_k(a)$ is nondecreasing in $a$, while $a^2=(\norm{z}^2-1/k)/(1-\norm{z}^2)$ is increasing in $\norm{z}^2$. Consequently, every clause in the stated range is satisfied by Gaussian rounding with probability at least $k^2/2^k$.

\paragraph{Randomized rounding.} As discussed above, Gaussian rounding can handle any clause such that $\norm{z}^2\geqslant(1+\eta_k)/k$, and for the rest we can simply output a uniformly random assignment. To execute such a strategy, it suffices to run Gaussian rounding with probability $1/k$, and output a uniformly random assignment with probability $1-1/k$. To see why, let $\alpha_k=\frac{1-1/k}{1+\eta_k}$ and observe,
if $\norm{z_C}^2\geqslant(1+\eta_k)/k$, the Gaussian rounding branch alone implies
\[
    \Pr[C\text{ is satisfied}]
    \geqslant
    \frac1k\cdot\frac{k^2}{2^k}
    \geqslant
    \alpha_k\frac{k}{2^k}\norm{z}^2.
\]
If $\norm{z_C}^2<(1+\eta_k)/k$, the random-assignment branch alone gives
\[
    \Pr[C\text{ is satisfied}]
    \geqslant
    \frac{1-1/k}{2^k} = \frac{\alpha_k(1+\eta_k)}{2^k}
    \geqslant
    \alpha_k\frac{k}{2^k}\norm{z_C}^2.
\]
Thus, every clause $C$ is satisfied with probability at least $\alpha_k(k/2^k)\norm{z_C}^2$. Summing over clauses yields an $\alpha_k \frac{k}{2^k}$ approximation. Finally, 
\[
    \alpha_k = \frac{1-1/k}{1+\eta_k} 
    =\frac{1-1/k}{1+(2\pi\log k+1)/k}
    =1-O\Paren{\frac{\log k}{k}},
\]
as desired. This argument already suffices to resolve Open Problem 2 of Makarychev and Makarychev~\cite{makarychev2017approximation}. We note that the $O(\log k/ k)$ loss in this argument comes from entirely from the low-score clauses, since the random assigment ignores the covariance structure entirely. Our main technical contribution is to remove this loss by proving a comparison inequality that remains useful below the critical threshold of Mulgund~\cite{Mulgund2026GaussianMaxima}. In the subsequent section, we state and prove this inequality, and in \cref{sec:main-proof} we give a formal proof of \cref{thm:main}.

\section{A Gaussian Orthant Amplification Inequality}\label{sec:gaussian}

For independent standard Gaussians, the positive orthant has probability exactly $2^{-m}$. We show that for Gaussians with correlation matrix $R$, the improvement is controlled by the precision of the all-ones direction, i.e. $\one^\top R^{-1} \one$. In particular, we show that when the all ones direction has lower precision than the independent case, the orthan probability is amplified by an exponential factor.

\begin{theorem}[Gaussian orthant amplification]\label{thm:zero-comparison}
Let $R\succ0$ be an $m\times m$ correlation matrix and let $X\sim\Normal(0,R)$. Then
\begin{equation}\label{eq:zero-comparison}
    \Pr[X\geqslant0]
    \geqslant
    2^{-m}\exp\Paren{
        \frac{m-\one^{\mathsf T}R^{-1}\one}{\pi}
    }.
\end{equation}
\end{theorem}

We require the following centered Gaussian product inequality, which captures the fact that once the log-concave functions are centered under the Gaussian measure, arbitrary correlations can only increase the joint expectation relative to the independent setting.

\begin{lemma}[Centered Gaussian product inequality~{\normalfont\cite[Theorem~4.1]{Mulgund2026GaussianMaxima}}]\label{lem:centered-product}
Let $R$ be an $m\times m$ correlation matrix and let $X\sim\Normal(0,R)$. Suppose that $f_1,\ldots,f_m\from\R\to[0,\infty)$ are bounded and log-concave and, for every $i\in[m]$, $\expecf{Z\sim\Normal(0,1)}{f_i(Z)}>0,$ and $\expecf{Z\sim\Normal(0,1)}{Zf_i(Z)}=0.$
Then,
\begin{equation*} 
    \expecf{X\sim\Normal(0,R)}{\prod_{i=1}^m f_i(X_i)}
    \geqslant
    \prod_{i=1}^m\expecf{Z\sim\Normal(0,1)}{f_i(Z)}.
\end{equation*}
\end{lemma}

We would like to apply this lemma to the indicators of the negative half-line since we want to lower bound $\Pr[X\leq0] = \E \prod_{i \in [n]} \mathbf{1}_{x_i\leq 0}$. These indicators, are log-concave but not centered under the Gaussian measure, so \cref{lem:centered-product} does not apply directly. Our proof follows Mulgund's high-level approach, but requiring a statement only for the positive orthant permits a more fine-grained analysis. We begin by replacing $\mathbf{1}_{x_i\leq 0}$ with $e^{ax} \mathbf{1}_{x_i\leq b_i }$ and show that there is a unique choice of $a_i$ such that $e^{ax} \mathbf{1}_{x_i\leq b_i }$ has zero mean under the Gaussian measure. The tilted distribution now corresponds to $x\sim\calN(Ra , R)$ and normalizing back to a distribution we have that 
\begin{equation*}
    e^{-a^{\top}R a/2} \expecf{}{e^{a^\top x} \mathbf{1}_{x\leq b}} = \Pr_{x\sim \calN(Ra, R)}[ x\leq b] =  \Pr_{x\sim \calN(0, R)}[ x\leq b - Ra]\,,
\end{equation*}
and thus we require $b=Ra$, which results in a variational problem. We show that it suffices to pick a test point with each $b_i$ being the same, i.e. $b = \sqrt{2/\pi} \mathbf{1}$.

\begin{proof}[Proof of \cref{thm:zero-comparison}]
By symmetry, it suffices to lower bound $\Pr[X\leqslant0]$.
We first show that, for every $b>0$, there is a unique $a(b)>0$ satisfying $\int_{-\infty}^b x e^{a(b)x}\,\diff\gamma_1(x)=0$. Indeed, the left-hand side is strictly increasing as a function of $a$, since its derivative is $\int_{-\infty}^b x^2e^{ax}\,\diff\gamma_1(x)>0$. At $a=0$, it equals $-\Phi'(b)<0$. As $a\to\infty$, the positive contribution from any interval contained in $(0,b]$ grows, while the negative contribution vanishes. Hence, there is a unique positive zero $a(b)$. 
Next, we define
\[
    F(b)
    =
    \log\int_{-\infty}^b e^{a(b)x}\,\diff\gamma_1(x).
\]
The quantity $e^{F(b)}$ is the Gaussian mass of the centered tilt $e^{a(b)x}\boldsymbol{1}_{\{x\leqslant b\}}$.
Let $b\in(0,\infty)^m$ be an arbitrary vector, and let $a(b)=(a(b_1),\ldots,a(b_m))$ coordinatewise. For notational convenience, let $a=a(b)$. The functions $f_i(x)=e^{a_i x}\boldsymbol{1}_{\{x\leqslant b_i\}}$ are bounded and log-concave, and their Gaussian first moments are centered by construction. Therefore, \cref{lem:centered-product} gives
\begin{equation}\label{eq:centered-product-application}
    \E\Bracks{
        e^{a^{\mathsf T}X}
        \boldsymbol{1}_{\{X\leqslant b\}}
    }
    \geqslant
    \exp\Paren{\sum_{i=1}^mF(b_i)}.
\end{equation}
Exponential tilting by $a$ shifts the mean of $\Normal(0,R)$ from $0$ to $Ra$. Combining this identity with \eqref{eq:centered-product-application}, we obtain
\begin{equation}\label{eq:shifted-orthant-bound}
    \Pr[X\leqslant b-Ra] =  e^{-\frac{1}{2} a^\top R a}\E\Bracks{e^{a^{\mathsf T}X}\boldsymbol{1}_{\{X\leqslant b\}}}
    \geqslant
    \exp\Paren{
        \sum_{i=1}^mF(b_i)
        -\frac12a^{\mathsf T}Ra
    }.
\end{equation}
Thus, we recover the desired orthant whenever $b=Ra$. It remains to find a vector satisfying this fixed-point condition. We do this by expressing the condition as a first-order optimality statement. Consider the potential
\begin{equation}\label{eq:orthant-variational}
    \Psi(b)
    =
    \sum_{i=1}^mF(b_i)
    -\frac12b^{\mathsf T}R^{-1}b,
    \qquad
    b\in(0,\infty)^m.
\end{equation}
The reason for this choice is that its gradient encodes the fixed-point condition. To compute it, let $I(b)=\int_{-\infty}^b e^{a(b)x}\,\diff\gamma_1(x)$, so that $F(b)=\log I(b)$. The implicit function theorem shows that $a(b)$ is continuously differentiable. Integration by parts and the centering identity give $e^{a(b)b}\Phi'(b)=a(b)I(b)$. Moreover, when differentiating $I(b)$, the term involving $a'(b)$ vanishes by the centering identity. Therefore, $F'(b)=a(b)$ and $\nabla\Psi(b)=a(b)-R^{-1}b$. Thus, every interior maximizer of $\Psi$ satisfies $b=Ra(b)$.
We now show that such a maximizer exists, i.e. $\Psi$ attains its maximum at a point $b^\star\in(0,\infty)^m$, the interior of the nonnegative orthant $[0,\infty)^m$. This requires ruling out maximizing sequences that either escape to infinity or approach a coordinate hyperplane.
We first rule out a maximizing sequence escaping to infinity. For fixed $b>0$, the function $t\mapsto\int_{-\infty}^b e^{tx}\,\diff\gamma_1(x)$ is strictly convex, and the centering identity says that its derivative vanishes at $t=a(b)$. Hence, $e^{F(b)}
    \leqslant
    \int_{-\infty}^b\diff\gamma_1(x)
    =
    \Phi(b)
    <1,$ and $F(b)<0$, which yields
\[
    \Psi(b)
    \leqslant
    -\frac12b^{\mathsf T}R^{-1}b
    \longrightarrow-\infty
    \qquad\text{as}\qquad
    \norm{b}\longrightarrow\infty.
\]

It remains to rule out the possibility that a maximizing sequence has some coordinate tending to zero.
Fix one coordinate $b>0$. The centering identity says that the probability measure proportional to $e^{a(b)x}\boldsymbol{1}_{\{x\leqslant b\}}\,\diff\gamma_1(x)$ has mean zero. As $b\to 0$, its positive support $(0,b]$ disappears, so the measure must be tilted increasingly strongly to the right, i.e. $a(b)\to\infty$. 
We would like to conclude that the total mass $e^{F(b)}=I(b)$ tends to zero. The integration-by-parts identity used above gives $ e^{F(b)}= \frac{e^{a(b)b}\Phi'(b)}{a(b)}$.
The denominator tends to infinity, but we must still rule out the possibility that the exponential term in the numerator compensates for it. This is precisely the role of a second integration by parts. Let $\nu_b$ be the probability measure obtained by normalizing the centered tilt. Since $\nu_b$ has mean zero,
\[
    \Var_{\nu_b}(X)
    =
    \frac{1}{I(b)}
    \int_{-\infty}^b x^2e^{a(b)x}\,\diff\gamma_1(x)
    =
    1-ba(b).
\]
The variance is always positive, and hence $ba(b)<1$, which inturn implies $e^{a(b)b}\leqslant e$, and returning to the preceding identity gives $e^{F(b)}
    \leqslant
    \frac{e\Phi'(0)}{a(b)}
    \longrightarrow0.$ 
Therefore, $F(b)\to-\infty$ as $b\to 0$, and no maximizing vector can have any coordinate tending to zero.
These two estimates imply that $\Psi$ attains its maximum in the interior. 
Indeed, the superlevel set $\Braces{b\in(0,\infty)^m:\Psi(b)\geqslant\Psi(\one)}$ is nonempty and, for some $0<\delta<M<\infty$, is a closed subset of $[\delta,M]^m$. It is therefore compact, so $\Psi$ attains its maximum there and hence over all of $(0,\infty)^m$. 
Let $b^\star$ be any maximizer.
Since $b^\star$ lies in the interior, its gradient vanishes and we have $b^\star=Ra(b^\star)$.
Applying \eqref{eq:shifted-orthant-bound} to $b^\star$, the event on the left becomes $X\leqslant0$ as desired, and $a(b^\star)^{\mathsf T}Ra(b^\star)=(b^\star)^{\mathsf T}R^{-1}b^\star$. Thus,
\begin{align}
    \log\Pr[X\leqslant0]
    &\geqslant
    \Psi(b^\star)=
    \max_{b\in(0,\infty)^m}
    \Braces{
        \sum_{i=1}^mF(b_i)
        -\frac12b^{\mathsf T}R^{-1}b
    }.
    \label{eq:orthant-variational-bound}
\end{align}
We do not need to compute $b^\star$ explicitly, since any test point lower bounds the maximum in \eqref{eq:orthant-variational-bound}. We use $b=\mu\one$, where $\mu=\sqrt{2/\pi}$. To evaluate $\Psi$ at this point, we first compute $a(\mu)$ and $F(\mu)$.

Recall that $\Phi(0)=1/2$ and $\Phi'(0)=1/\sqrt{2\pi}=\mu/2$. Completing the square and changing variables gives
\[
    \int_{-\infty}^{\mu}x e^{\mu x}\,\diff\gamma_1(x)
    =
    e^{\mu^2/2}\Paren{-\Phi'(0)+\mu\Phi(0)}
    =0.
\]
This is precisely the centering identity defining $a(\mu)$, and hence uniqueness gives $a(\mu)=\mu$. Using this value and completing the square once more,
\[
    e^{F(\mu)}
    =
    \int_{-\infty}^{\mu}e^{\mu x}\,\diff\gamma_1(x)
    =
    e^{\mu^2/2}\Phi(0)
    =
    \frac12e^{\mu^2/2}.
\]
Therefore, $F(\mu)=-\log2+\mu^2/2=-\log2+1/\pi$. Substituting $b=\mu\one$ into \eqref{eq:orthant-variational-bound}, we obtain
\begin{align*}
    \log\Pr[X\leqslant0]
    \geqslant \Psi(\mu\one)=
    mF(\mu)
    -\frac{\mu^2}{2}\one^{\mathsf T}R^{-1}\one=
    -m\log2
    +\frac{m-\one^{\mathsf T}R^{-1}\one}{\pi}.
\end{align*}
Exponentiating and using the symmetry of $X$ gives
\[
    \Pr[X\geqslant0]
    =
    \Pr[X\leqslant0]
    \geqslant
    2^{-m}\exp\Paren{
        \frac{m-\one^{\mathsf T}R^{-1}\one}{\pi}
    },
\]
which proves the claim.
\end{proof}

The following implication is the form we use in the rounding analysis.

\begin{corollary}[Rank-one orthant amplification]\label{cor:rank-one-orthant}
Let $0<\eta\leqslant1$, let $R$ be an $m\times m$ correlation matrix satisfying $R-\eta J\succeq0$, and let $X\sim\Normal(0,R)$. Then
\begin{equation}\label{eq:rank-one-orthant}
    \Pr[X\geqslant0]
    \geqslant
    2^{-m}\exp\Paren{\frac{m-\eta^{-1}}{\pi}}.
\end{equation}
\end{corollary}

\begin{proof}
Suppose first that $R\succ0$. Since $R-\eta J\succeq0$,
\[
    0
    \leqslant
    \one^{\mathsf T}R^{-1}(R-\eta J)R^{-1}\one
    =
    \one^{\mathsf T}R^{-1}\one
    -\eta\Paren{\one^{\mathsf T}R^{-1}\one}^2,
\]
and hence $\one^{\mathsf T}R^{-1}\one\leqslant\eta^{-1}$. The claim now follows from \cref{thm:zero-comparison}.
For singular $R$, set $R_\varepsilon=(1-\varepsilon)R+\varepsilon I$ and $\eta_\varepsilon=(1-\varepsilon)\eta$. Then $R_\varepsilon-\eta_\varepsilon J\succeq0$, so the nonsingular case applies. Letting $\varepsilon\to0$ completes the proof.
\end{proof}

\section{Algorithm and Analysis}\label{sec:main-proof}

We now prove the main theorem using the Gaussian orthant amplification inequality from the previous section. Formally,

\begin{theorem}[Boolean Max-$k$-CSP]\label{thm:main-restated}
For every fixed integer $k\geqslant10$ and every $\varepsilon>0$, there is a randomized algorithm that, given any Boolean Max-$k$-CSP instance $\calI$ on $n$ variables, with $m$ constraints, outputs an assignment $x\in\{0,1\}^n$ satisfying
\[
    \E[\val_{\calI}(x)]
    \geqslant
    \Paren{\frac{k}{2^k}-\varepsilon}\OPT(\calI),
\]
and runs in $\poly(n, m,k, \log(1/\varepsilon))$ time.
\end{theorem}

\begin{mdframed}
\begin{algorithm}[Noisy Gaussian rounding]\label{alg:noisy-rounding}\mbox{}
\begin{description}
\item[Input:] A set $\mathcal C$ of weighted $k$-literal conjunctions. Each $C\in\mathcal C$ is specified by a support $\supp(C)\subseteq[n]$ of size $k$, a required assignment $b_C\in\zo^{\supp(C)}$, and a weight $\lambda_C\geqslant0$. The conjunction $C$ is satisfied exactly when $x_i=b_C(i)$ for every $i\in\supp(C)$.

\item[Operations:]\mbox{}
\begin{enumerate}
    \item Compute an optimal solution to the following SDP:
    \begin{align*}
        \textup{maximize}\qquad
            &\sum_{C\in\mathcal C}\lambda_C\norm{z_C}^2 \\
        \textup{subject to}\qquad
            &\norm{v_{i,0}}^2+\norm{v_{i,1}}^2\leqslant1
                &&\text{for every }i\in[n],\\
            &\langle v_{i,0},v_{i,1}\rangle=0
                &&\text{for every }i\in[n],\\
            &\langle v_{i,b_C(i)},z_C\rangle=\norm{z_C}^2
                &&\text{for every }C\in\mathcal C\text{ and }i\in\supp(C),\\
            &\langle v_{i,1-b_C(i)},z_C\rangle=0
                &&\text{for every }C\in\mathcal C\text{ and }i\in\supp(C),\\
            &\norm{z_C}^2\leqslant1
                &&\text{for every }C\in\mathcal C.
    \end{align*}

    \item Let $d_i=v_{i,1}-v_{i,0}$. The SDP constraints imply $\norm{d_i}\leqslant1$. Add a component in a fresh direction to make each $d_i$ a unit vector, choosing these directions mutually orthogonal and orthogonal to every SDP vector. Continue to denote the resulting vector by $d_i$.

    \item Let $\delta=\pi/k$. Sample a standard Gaussian vector $g$ in the extended SDP vector space and independent standard Gaussians $\xi_1,\ldots,\xi_n$. Let
    \begin{equation*}
        x_i
        =
        \boldsymbol{1}\Braces{
            \sqrt{\delta}\,\langle g,d_i\rangle
            +\sqrt{1-\delta}\,\xi_i
            \geqslant0
        }.
    \end{equation*}
\end{enumerate}

\item[Output:] The assignment $x\in\zo^n$.
\end{description}
\end{algorithm}
\end{mdframed}

We require the following scalar inequality, which shows that the function $s\mapsto\exp\Paren{ks/(k+\pi s)}$ lies above its tangent line $1+s$ throughout the interval $[-1,k-1]$.

\begin{lemma}[Logarithmic-mean bound]\label{lem:logarithmic-mean}
For every integer $k\geqslant10$ and every $s\in[-1,k-1]$,
\[
    \exp\Paren{\frac{ks}{k+\pi s}}
    \geqslant
    1+s.
\]
\end{lemma}

\begin{proof}
Let $\alpha=\pi/k$ and $t=1+s\in[0,k]$. The case $t=0$ is immediate. For $t>0$, consider the logarithmic mean
\[
    L(t)
    =
    \frac{t-1}{\log t}
    =
    \int_0^1t^u\,\diff u,
    \qquad
    L(1)=1.
\]
For every $u\in[0,1]$, the function $t\mapsto t^u$ is concave, and hence $L$ is concave. First, we consider the case where $0<t\leqslant1$. Differentiating the integral at $t=1$ gives $L'(1)=1/2$. Concavity bounds $L$ from above by its tangent line at $1$, and hence
\[
    L(t)
    \leqslant
    1+\frac{t-1}{2}
    \leqslant
    1+\alpha(t-1),
\]
where the second inequality uses $\alpha\leqslant1/2$ and $t-1\leqslant0$. It  then follows that
\[
    \log t
    =
    \frac{t-1}{L(t)}
    \leqslant
    \frac{t-1}{1+\alpha(t-1)} = \frac{sk}{k+ \pi s}.
\]

Now consider the case where $1\leqslant t\leqslant k$. For this range, we first verify the inequality at the endpoint $k$, i.e.  $L(k)\geqslant1+\alpha(k-1)$. This is equivalent to $ \frac{\log k}{k} \leqslant \frac{k-1}{k+\pi(k-1)}.$ 
For $k\geqslant10$, the left-hand side is decreasing in $k$, while the right-hand side is increasing. It therefore suffices to check $k=10$. Using $\log 10<7/3$ and $\pi<22/7$, we have $\log (10)/10 <  9/(10+9\pi)$ as desired. 
Concavity now places $L(t)$ above the chord joining $1$ and $k$ and using that $L(k)\geqslant1+\alpha(k-1)$, we have
\[
    L(t)
    \geqslant
    \frac{k-t}{k-1}L(1)
    +\frac{t-1}{k-1}L(k)
    \geqslant
    1+\alpha(t-1).
\]
Since $t-1$ is nonnegative, this implies $\log t= \frac{t-1}{L(t)}\leqslant \frac{t-1}{1+\alpha(t-1)}$, and the same inequality holds throughout $t\in(0,k]$. Exponentiating and substituting $t=1+s$ and $\alpha=\pi/k$ proves the claim.
\end{proof}

We now combine the rank-one orthant amplification inequality with \cref{lem:logarithmic-mean} to prove \cref{thm:main-restated}. The analysis is clausewise. For a conjunction $C$, let $q=\norm{z_C}^2$ and $s=kq-1$. We show that the covariance matrix of its oriented Gaussian scores dominates $\eta J$, where $\eta=(1+\delta s)/k$. The rank-one orthant amplification inequality then gives
\[
    \Pr[C\text{ is satisfied}]
    \geqslant 2^{-k}
    \exp\Paren{\frac{k-\eta^{-1}}{\pi}}.
\]
Our choice $\delta=\pi/k$ makes the exponent $ks/(k+\pi s)$, and \cref{lem:logarithmic-mean} lower-bounds the right-hand side by $kq$. Thus, the rounding satisfies each conjunction with probability at least $(k/2^k)\norm{z_C}^2$. Summing these clausewise guarantees proves the theorem.

\begin{proof}[Proof of \cref{thm:main-restated}]

We begin by padding every predicate of arity less than $k$ with fresh variables to get clauses of arity exactly $k$, and then apply Trevisan's reduction from \cref{sec:overview}. Both operations preserve the value of every assignment, so it suffices to work with the resulting weighted $k$-literal conjunctions. We retain the notation $\lambda_C$, $\supp(C)$, and $b_C$ from the overview.
The integral embedding from \cref{sec:overview} shows that $\SDP\geqslant\OPT$. We now analyze an arbitrary conjunction $C\in\mathcal C$. Let $q=\norm{z_C}^2$. For each $i\in[n]$, let
\[
    s_i
    =
    \sqrt{\delta}\,\langle g,d_i\rangle
    +\sqrt{1-\delta}\,\xi_i
\]
denote the score used by the rounding, so that $x_i=\boldsymbol{1}\Braces{s_i\geqslant0}$. To express the event that $C$ is satisfied as a positive orthant event, we change the sign of this score whenever $C$ requires value $0$. Namely, for every $i\in\supp(C)$, let $ y_{C,i}=(2b_C(i)-1)s_i.$
If $b_C(i)=1$, then $y_{C,i}=s_i$, so a positive score assigns the required value $1$. If $b_C(i)=0$, then $y_{C,i}=-s_i$, so a positive oriented score means that $s_i$ is negative and the rounding assigns the required value $0$. Since $s_i$ is a nondegenerate Gaussian, it is zero with probability zero. Thus, almost surely, $y_{C,i}\geqslant0$ exactly when $x_i=b_C(i)$, and $C$ is satisfied exactly when $y_{C,i}\geqslant0$ for every $i\in\supp(C)$.

We next express the covariance of these oriented scores. Recall that $d_i=v_{i,1}-v_{i,0}$, and let $w_{C,i}=(2b_C(i)-1)d_i$. Thus, $w_{C,i}=d_i$ when $C$ requires value $1$ and $w_{C,i}=-d_i$ when it requires value $0$. We can then rewrite
\[
    y_{C,i}
    =
    \sqrt{\delta}\,\langle g,w_{C,i}\rangle
    +\sqrt{1-\delta}\,(2b_C(i)-1)\xi_i.
\]
Using the constraints that $\langle v_{i,b_C(i)},z_C\rangle=q$ and $\langle v_{i,1-b_C(i)},z_C\rangle=0$ we have
\[
    \langle w_{C,i},z_C\rangle
    =
    \langle v_{i,b_C(i)},z_C\rangle
    -\langle v_{i,1-b_C(i)},z_C\rangle
    =
    q\,.
\]
As in \cref{eq:overview-clause-covariance}, the Gram matrix of the $w_{C,i}$'s therefore has the form $qJ+(1-q)G_C$,
for some $k\times k$ correlation matrix $G_C$. Multiplying $\xi_i$ by $2b_C(i)-1$ does not change its distribution, so the oriented noises remain independent standard Gaussians. Consequently, $y_C=(y_{C,i})_{i\in\supp(C)}$ is a centered Gaussian with correlation matrix $R_C=\delta\Paren{qJ+(1-q)G_C} +(1-\delta)I,$
and $\Pr[C\text{ is satisfied}] = \Pr[y_C\geqslant0].$
Let $s=kq-1\in[-1,k-1],$ and let $\eta = \delta q+\frac{1-\delta}{k} = \frac{1+\delta s}{k}$. Then, we have
\begin{equation*}
     R_C-\eta J = \delta(1-q)G_C +(1-\delta)\Paren{I-\frac1kJ}\succeq0.
\end{equation*}
Indeed, $G_C\succeq0$, while $I-J/k$ is the orthogonal projector onto the subspace perpendicular to $\one$. Moreover, $R_C\succeq(1-\delta)I\succ0$ and $0<\eta \leqslant\delta+\frac{1-\delta}{k}<1.$
Thus, applying \cref{cor:rank-one-orthant} we have
\[
   \Pr[C\text{ is satisfied}]
    \geqslant 2^{-k}
    \exp\Paren{\frac{k-\eta^{-1}}{\pi}}.
\]
Since $\eta=(1+\delta s)/k$ and $\delta=\pi/k$, $\frac{k-\eta^{-1}}{\pi} = \frac{ks}{k+\pi s}$, and applying \cref{lem:logarithmic-mean}, we conclude that
\[
    \Pr[C\text{ is satisfied}]
    \geqslant 2^{-k}
    \exp\Paren{\frac{ks}{k+\pi s}}
    \geqslant
    2^{-k}\;(1+s)
    =
    k \; 2^{-k}\; \norm{z_C}^2,
\]
which concludes the proof.
This proves the guarantee in exact arithmetic. Using standard arguments for solving a SDP and implementing the Gaussian rounding procedure in finite precision, we can get the stated running time.
\end{proof}

\section*{Acknowledgements and AI disclosure}

The author first encountered this problem during a summer internship at TTI-Chicago in 2021 and is grateful to Yury Makarychev for several illuminating discussions, particularly for pointing out its connection to Gaussian orthant probabilities. The author has returned to the problem periodically since then and thanks Romain Cosson, Giada Franz, Jia Shi, and Jingze Zhu for helpful conversations. 

All the conceptual contributions in this work are biological. The author came across Mulgund's proof of the Weak Simplex Conjecture and realized the Gaussian comparison inequality could be used to obtain a $(1-o(1))k/2^k$ approximation. This argument appears in the technical overview in its essence. The author then formulated the orthant amplification statement that would suffice to get
an approximation ratio of $k/2^{k}$. GPT 5.6 Sol Max assisted in the lengthy computations that appear in the proof of \cref{thm:zero-comparison}, which sped up the completion time significantly. The author wrote the entire manuscript and takes full responsibility for its content and correctness.

\begingroup
\footnotesize
\printbibliography
\endgroup

\end{document}